\documentclass[runningheads]{llncs}
\usepackage[english]{babel}
\usepackage{float} %%
\usepackage{epic,eepic,latexsym,amssymb,color,amscd}
\usepackage{ifthen,graphics,epsfig,xspace,times}
\usepackage{xspace}
\usepackage{thmtools, thm-restate}
\usepackage{hyperref}
\usepackage{caption}
\usepackage{graphicx}
\usepackage{subcaption}
\usepackage{wrapfig}

\newlength {\squarewidth}

\newcounter{linecounter}
\newcommand{\linenumbering}{\ifthenelse{\value{linecounter}<10}
{(\arabic{linecounter})}{(\arabic{linecounter})}}
\renewcommand{\line}[1]{\refstepcounter{linecounter}\label{#1}\linenumbering}
\newcommand{\resetline}[1]{\setcounter{linecounter}{0}#1}
\renewcommand{\thelinecounter}{\ifnum \value{linecounter} > 
9 \else \fi\arabic{linecounter}}
\newfloat{algorithm}{thp}{lop}
\floatname{algorithm}{Algorithm}
\newfloat{construction}{thp}{lop}
\floatname{construction}{Construction}
\newcommand{\Xomit}[1]{}

\newcommand{\timeout}{\mathit{timeout}}
\newcommand{\timer}{\mathit{timer}}
\newcommand{\penal}{\mathit{penalty}}
\newcommand{\set}{{\sf{set}}}

\newcommand{\llog}{{\sf{log}}}
\newcommand{\mmax}{{\sf{max}}}
\newcommand{\send}{{\sf{send}}}
\newcommand{\tto}{{\sf{to}}}
\newcommand{\clock}{{\sf{clock}}}
\newcommand{\ALIVE}{{\sc{alive}}}
\newcommand{\correct}{{\sf{correct}}}
\newcommand{\crashed}{{\sf{crashed}}}

\newcommand{\new}{{\tt{new}}}
\newcommand{\ack}{{\tt{ack}}}

\begin{document}
\title{Leader Election in  Arbitrarily Connected Networks
with  Process Crashes and Weak Channel Reliability
%\thanks{Partially supported by UNAM-PAPIIT grant  IN106520.
\\
}
\author{Carlos López$^{\dag}$,Sergio Rajsbaum$^{\dag}$, 
    Michel Raynal$^{\star,\ddag}$,
    Karla Vargas$^{\dag}$ ~\\~\\
$^{\dag}$Instituto de Matem\'aticas, UNAM, Mexico \\
$^{\star}$Univ Rennes IRISA, 35042 Rennes, France \\
$^{\ddag}$Department of Computing, Polytechnic University, Hong Kong
}
\institute{}
\authorrunning{C.López, S. Rajsbaum, M. Raynal, K. Vargas}

\titlerunning{Leader election with  Process Crashes and Weak Channel Reliability
 }
\maketitle              % typeset the header of the contribution
\vspace{-.3cm}
\begin{abstract}
  A channel from a process $p$ to a process $q$ satisfies the
  \emph{ADD property} if there are constants $K$ and $D$, unknown to
  the processes, such that in any sequence of $K$ consecutive messages
  sent by $p$ to $q$, at least one of them is delivered to $q$ at most
  $D$ time units after it has been sent.
  This paper studies implementations of an eventual leader,
  namely, an $\Omega$ failure detector, in an arbitrarily connected network of eventual ADD channels, where
  processes may fail by crashing.  It first presents an algorithm that
  assumes that processes initially know $n$, the total number of
  processes, sending messages of size $O( ~\llog~ n)$.  Then, it presents
  a second algorithm that does not assume the processes know $n$.
  Eventually the size of the messages sent by this algorithm is also
  $O( ~\llog~ n)$.  These are the first implementations of leader
  election in the ADD model.  In this model, only eventually perfect
  failure detectors were considered, sending messages of size 
  $O(n  ~\llog~ n)$.  ~\\
  
 {\bf Keywords}:
ADD channel, Arbitrarily Connected Networks, Distributed algorithm, Eventual leader, Fault-tolerance,
 Process crash, Synchrony, System model, Unknown membership, Weak channel.
\end{abstract}
\vspace{-0.7cm}

%========================================================================
\section{Introduction}
\vspace{-0.2cm}
\subsection{Leader election}
\vspace{-0.2cm}

This is a classical problem encountered in distributed computing.
Each process $p_i$ has a local variable $leader_i$, and it is required
that all the local variables $leader_i$ forever contain the same
identity, which is the identity of one of the processes.  A classical
way to elect a leader consists in selecting the process with the
smallest identity\footnote{A survey on election algorithms in
  failure-free message-passing systems appears in Chapter 4
  of~\cite{R13}. The aim is to elect a leader as soon as possible, and
  with as few messages as possible, and it can be done on a ring with
  $1.271 ~n ~\llog(n)+ O(n)$ messages~\cite{HP96,MAR12}.}.  If
processes may crash, the system is fully asynchronous, and the elected
leader must be a process that does not crash, leader election cannot
be solved~\cite{R18}. Not only the system must no longer be fully
asynchronous, but the leader election problem must be weakened to the
{\it eventual leader election problem}. This problem is denoted
$\Omega$ in the failure detector parlance~\cite{CHT96,CT96}. Notice
that the algorithm must elect a new leader each time the previously
elected leader crashes.

\vspace{-.3cm}
\subsection{Related work}
Many algorithms for electing an eventual leader in crash-prone
partially synchronous systems have been proposed. Surveys of such algorithms
are presented in~\cite[Chapter 17]{R13-cp} when communication is
through a shared memory, and in~\cite[Chapter 18]{R18} when
communication is through reliable message-passing. 

In~\cite{ADFT04} there are proposed different levels of communication reliability and is its showed that in systems with only some timely channels and a complete network it is necessary that correct processes send messages forever even with just at most one process crash. An algorithm for implementing $\Omega$ in networks with unknown membership is presented in ~\cite{JAF06}. This algorithm works in a complete network and every process needs to communicate its name to every neighbor using a broadcast protocol.

In~\cite{FLCR17} it is presented an implementation of $\Omega$ for the case of the crash-recovery model in which processes can crash and then recover infinitely many times and channels can lose messages arbitrarily. The case of dynamic systems is addressed in~\cite{LSCR12}, and the case where the
underlying synchrony assumptions may change with time is addressed
in~\cite{FR10}. Stabilizing  leader election in crash-prone
synchronous systems is investigated in~\cite{DDF10}. 

\vspace{-0.2cm}
%-----------------------------------------------------------------------
\subsubsection{The ADD distributed computing model}

\vspace{-0.1cm}
This model was introduced in ~\cite{SP07}, as a realistic partially
synchronous model of channels that can lose and reorder
messages\footnote{An early look at unreliable channel properties to
  design a ``fail-aware'' service appeared in~\cite{F96}.}.  Each
channel guarantees that some subset of the messages sent on it will be
delivered in a timely manner and such messages are not too sparsely
distributed in time.  More precisely, for each channel there exist two
constants $K$ and $D$, not known to the processes (and not necessarily
the same for all channels), such that for every $K$ consecutive
messages sent in one direction, at least one is delivered within $D$
time units after it has been sent.

Even though ADD channels seem so weak, it is possible to implement an
\textit{eventually perfect failure detector}, $\Diamond P$, in a fully
connected network of ADD channels, where asynchronous processes may
fail by crashing~\cite{SP07}.  Later on, it was shown that it is also
possible to implement $\Diamond P$ in an arbitrarily connected network
of ADD channels~\cite{KW19}.  Recall that $\Diamond P$ is a classic
failure detector, relatively powerful (more than sufficient to solve
consensus), stronger than $\Omega$~\cite{CHT96} and yet realistically
implementable~\cite{CT96}.

The algorithm in~\cite{KW19} works for arbitrary connected networks of
ADD channels, and sends messages of bounded size, improving on the
previous unbounded size algorithm presented in~\cite{H04}.  However,
the size of messages is exponential in $n$, the number of processes.
More recently, an implementation of $\Diamond P$ using messages of
size $ O(n~\llog~n)$, in an arbitrarily connected network of ADD
channels was presented in~\cite{VR19}.
 \vspace{-0.3cm}

%-----------------------------------------------------------------------
 \subsection{Contribution}
 \vspace{-0.1cm}
 
 This paper shows that it is possible to implement $\Omega$ in an
 arbitrarily connected network where asynchronous processes may fail
 by crashing in a weaker model than the one presented in~\cite{VR19}. It first presents an implementation of $\Omega$ with messages of size
 $O(\llog~n)$, reducing the message size with respect to~\cite{VR19}.
 
  Most of the previous works related to $\Omega$  concentrated on communication-efficient algorithms  in fully connected networks when considering the number of messages. They were considering neither the size of the messages nor arbitrarily connected networks.
 
 The proposed algorithm works under very weak assumptions, requiring only that a directed
 spanning tree from the leader exists, composed of channels that
 eventually satisfy the ADD property.  This algorithm requires that
 processes know $n$, the number of processes.  Then, the paper shows
 how to extend the ideas to design an algorithm for the case where $n$
 is unknown in arbitrarily connected networks.  Initially a process knows only its set of incident
 channels.  Interestingly enough, eventually the size of the messages
 used by this algorithm is also $O(\llog~n)$.
 
We put particular attention to the size of the messages because it
plays an important role in the time it takes for the processes to
agree on the same leader, yet we show that our algorithms elect a
leader in essentially optimal time.  When designing ADD-based
algorithms, it is challenging to transmit a large message by splitting
it into smaller messages, due to the uncertainty created by the fact
that, while the constants $K$ and $D$ do exist, a process knows
neither them nor the time from which the channels forever
satisfy them.  This
type of difficulty is also encountered in the design of leader
election algorithms under weak eventual synchrony assumptions,
e.g.,~\cite{ADFT04,FJRT10,VR19}.  Also in self-stabilizing problems,
where ideas similar to our hopbound technique have been
used~\cite{DDP19}, as well as in~\cite{VR19}. 
% ADD channels can delay
% messages over the bound $D$ or even lose messages, infinitely often,
% while \emph{eventually timely} channels assume there is a duration $D$
% and a time after which every message sent to a correct process is
% received within $D$ time units after it has been sent.  Namely,
% additionally to eventually timely channels, papers such
% as~\cite{ADFT04} assume a complete network. 
We found it even more
challenging to work under the assumption that some edges might not
satisfy any property at all; our algorithm works under the assumption
that only edges on an (unknown to the processes) spanning tree are
guaranteed to comply with the ADD property.
\vspace{-0.5cm}

% \subsubsection{Organization}
% In Section~\ref{sec:model} we describe the model of computation.
% In Section~\ref{sec:algorithm} we describe our main algorithm, and its correctness proof is in
% Section~\ref{sec:proof}.
% In Section~\ref{sec:unknown-membership} we extend it to work with unknown membership, and the
% correctness proof is in 
%  Section~\ref{sec:proofUknown}.
% Section~\ref{sec:concl} concludes the paper, and most proofs appear in the Appendix.

%====================================================================
\section{Model of Computation}\label{sec:model}
\vspace{-0.2cm}
\paragraph{Process model}
The system consists of a finite set of $n$ processes $\Pi =
\{p_1,p_2,...,p_n\}$.  Every process $p_i$ has an identity, and
without loss of generality we consider that the identity of $p_i$ is
its index $i$.  As there is no ambiguity, we use indifferently $p_i$
or $i$ to denote the same process.

Every process $p_i$ has also a read-only local clock  $\clock_i()$, which
 is assumed to generate ticks at a constant 
  %SR I don't think we use this assumption Karla?
% (the same for all the processes)
 rate\footnote{However, the algorithm presented below
  can be adapted in the case where local clocks can suffer bounded
  drifts, these drifts being known by the processes.}. 
  %SR: I don't think we need this, Karla?
%  While they progress at the same speed, 
Local clocks need not to be synchronized
to exhibit the same time value at the same time,
local clocks are used only to implement timers. 
%SR: shorter, and classic way of doing it
To simplify the presentation,  it is assumed that local computations have
zero duration.
% (this is motivated by the observation that the duration of a local computation is negligible when compared to communication delays).

Any number of processes may fail by crashing.
%SR: shorter, and classic way of doing it
% (where a crash is a premature stop). 
%According to usual terminology and  Given an execution,  
A process is {\it correct} if it does not crash, 
otherwise, it is {\it faulty}. 
% (this is motivated by the observation that the duration of a local computation is negligible when compared to communication delays).
\vspace{-.2cm}

\paragraph{Virtual global clock}
For notational simplicity, 
%in the explanation and the proofs,
we assume the existence of an external reference clock which
remains always unknown to the processes. The range of its ticks 
is the set of natural numbers. It allows to associate
consistent dates with events generated by the algorithm.\\
\vspace{-.5cm}

\paragraph{Communication network}
It is represented by a directed graph
$G=(\Pi,E)$, where an edge $(p_i,p_j)\in E$ means that there is a
unidirectional channel that allows the process $p_i$ to send messages
to $p_j$. 
%SR: phrasing and shorter, removing well-know footnote
A bidirectional channel can be represented by two
unidirectional channels, possibly with different timing assumptions.
%this assumption is very general\footnote{Examples
 %of such networks are radio-networks where two neighbor stations do not
 %necessarily have the same transmit power.}. It follows Thus, each
process $p_i$ has a set of input channels and a set of output
channels.

The graph connectivity requirement on the communication graph $G$
depends on the problem to be solved. It will be stated in
the section~\ref{sec:proof} and~\ref{sec:unknown-membership}
devoted to the proofs of the proposed algorithms. 
\vspace{-.2cm}

\paragraph{Basic channel property}
It is assumed that no directed channel
creates, corrupts,  or duplicates messages. 

\vspace{-.3cm}

\paragraph{The ADD   property}
%SR removed  liveness for consistency later on
A directed channel $(p_i,p_j)$ satisfies the \emph{ADD property} if
  there are two constants $K$ and $D$ (unknown to the
  processes\footnote{Always unknown, as the global time, also never
    known by the processes.})  such that
\begin{itemize}
\item for every $K$ consecutive messages sent by $p_i$ to $p_j$, at
  least one is delivered to $p_j$ within $D$ time units after it has
  been sent.  The other messages from $p_i$ to $p_j$ can be lost or
  experience arbitrary delays.
\end{itemize}
Each directed channel can have its own pair $(K,D)$.
To simplify the presentation, and without loss of generality,
we assume that the pair $(K,D)$ is the same for all the channels. 

\vspace{-.3cm}

\paragraph{The $\Diamond$ADD property}
The eventual ADD property, states that the ADD property is satisfied
only after an unknown but finite period of time.  Hence this weakened
property allows the system to experience an initial anarchy period
during which the behavior of the channels is arbitrary.

\vspace{-.3cm}

\paragraph{The Span-Tree assumption}

We consider that there is a time $\tau$ after which
there is a directed spanning tree (i) that includes all the correct
processes and only them, (ii) its root is the correct process with the
smallest identity, and (iii) its channels satisfy the $\Diamond$ADD
property. This behavioral assumption is called {\it Span-Tree} in the
following.
 
 \vspace{-.3cm}

\paragraph{Eventual leader election}
Assuming a read-only local variable $leader_i$ at each process $p_i$,
the leader failure detector $\Omega$ satisfies the following
properties~\cite{CHT96,R05}:
\begin{itemize}
\item \emph{Validity:} For any process $p_i$,  each
  read of $leader_i$ by $p_i$ returns a process
   name.
 \item \emph{Eventual leadership:}
   There is a finite (but unknown) time after which the local variables
   $leader_i$ of all the correct processes contain forever the same process
   name,   which is the name of one of them.
\end{itemize}
\vspace{-.4cm}

%========================================================================
\section{Eventual Leader Election with Known Membership}
\label{sec:algorithm}
\vspace{-.2cm}

%SR: I think we don't need it, Karla?
This section presents Algorithm~\ref{algo:description} that implements
$\Omega$, assuming each process knows $n$. Parameter $T$ denotes an arbitrary duration. Its value affects the
efficiency of the algorithm, but not its correctness\footnote{If $T$
  is too big, the failure detection of a process currently considered
  as a leader can be delayed. On the contrary, a too small value of
  $T$ can entail false suspicions of the current eventual leader $p_j$
  until the corresponding timer $\timer_i[j]$ has been increased to an
  appropriate timeout value.}.
\vspace{-.3cm}

%-------------------------------------------------------------------

\subsection{Local variables at a process $p_i$}
\vspace{-.1cm}
Each process $p_i$ manages the following local variables.
\begin{itemize}
\item $in\_neighbors_i$ (resp., $out\_neighbors_i$) is 
 a (constant) set containing the identities of the processes $p_j$ 
 such that there is channel from $p_j$ to $p_i$
 (resp., there is channel from $p_i$ to $p_j$).  
\item $leader_i$ contains the identity of the elected leader. 
\item $\timeout_i[1..n, 1..n]$ is a matrix of timeout values and
 $\timer_i[1..n, 1..n]$ is a matrix of timers, such that the pair
 $\langle \timer_i[j,n-k],\timeout_i[j,n-k]\rangle$ is used by $p_i$
 to monitor the elementary paths from $p_j$ to it whose length is $k$.
\item $hopbound_i[1..n]$ is an array of non-negative integers;
 $hopbound_i[i]$ is initialized to $n$, while each other entry
 $hopbound_i[j]$ is initialized to $0$. Then, when $j\neq i$,
 $hopbound_i[j]=n-k \neq 0$ means that, if $p_j$ is currently
 considered as leader by $p_i$, the information carried by the last
 message \ALIVE$(j,n-1)$ sent by $p_j$ to its out-neighbors (which
 forwarded \ALIVE$(j,n-2)$ to their out-neighbors, etc.) went through
 a path\footnote{In the graph theory, such a cycle-free
  path is called an {\it elementary} path.} of $k$ different
 processes before being received by $p_i$. The code executed by $p_i$
 when it receives a message \ALIVE$(j,-)$ is  detailed in
 Section~\ref{sec:behavior}.

 The identifier $hopbound$ stands
 for ``upper bound on the number of forwarding'' that --due to the
 last message \ALIVE$(j,-)$ received by $p_i$-- the message
 \ALIVE$(j,-)$ sent by $p_i$ has to undergo to be received by all
 processes. It is similar to a {\it time-to-live} value.  
\item $penalty_i[1..n,1..n]$ is a matrix of integers such that $p_i$
 increases $penalty_i[j,n-k]$ each time the $timer_i[j,n-k]$
 expires. It is a penalization counter monitored by $p_i$ with respect
 to the elementary paths of length $k$ starting at $p_j$ and ending at $p_i$.  
\item $not\_expired_i$ is an auxiliary local variable. 
\end{itemize}

\vspace{-.3cm}

\subsection{General principle of the algorithm}

% \subsection{Detailed behavior of a process $p_i$}

\vspace{-.1cm}
\label{sec:behavior}
As many other leader election algorithms, Algorithm~\ref{algo:description}
elects the process that has the smallest identity among the set of
correct processes by keeping as a candidate to be the leader the smallest identifier received as it is explained in the following sections. It is made up of  three main sections:
the one that generates and forwards the \ALIVE$()$ messages, the one
that receives \ALIVE$()$ messages and the one that handles the
timer expiration. Every section is described in detail below.
%\vspace{-.1cm}

\vspace{-.5cm}
\subsubsection{Generating and forwarding messages}(Lines~\ref{ADD-06}-\ref{ADD-09})
Every $T$ time units of clock $\clock_i()$, a process $p_i$ sends the
message \ALIVE$(leader_i,hopbound_i[leader_i]-1)$.

A message \ALIVE$(*,n-1)$ is called {\it generating} message. A message \ALIVE$(*,n-k)$ such that $1 < k < n-1$, 
is called {\it forwarding} message (in this case
it is the forwarding of the last message \ALIVE$(leader_i,hopbound)$
 previously received by $p_i$). Moreover, the value
 $n-k$ is called \emph{hopbound value}. When a process $p_i$ starts the algorithm, it proposes itself as
candidate to be leader. 

A message is sent if predicate $hopbound_i[leader_i]>1$ of line~\ref{ADD-07} is true,
% The message sending is controlled by the predicate of
% line~\ref{ADD-07}, 
 The
 message sent is then \ALIVE$(leader_i,hopbound_i[leader_i]-1)$.

The message forwarding is motivated by the fact that, if
$hopbound_i[leader_i] >1$, maybe processes have not yet received a
message \ALIVE$(leader_i,-)$ whose sending was initiated by $leader_i$
and then forwarded along paths of processes (each process having
decreased the carried hopbound value) has not reached all the processes. In
this case, $p_i$ must participate in the forwarding. To this end, it
sends the
message \ALIVE$(leader_i,hopbound_i[leader_i]-1)$ to each of its
out-neighbors (line~\ref{ADD-08}).

Let us observe that during the anarchy period during which, due to the
values of the timeouts and the current asynchrony, channel behavior
and process failure pattern, several generating messages
\ALIVE$(*,n-1)$ can be sent by distinct processes (which compete to
become leader) and forwarded by the other processes with decreasing
hopbound values. But, when there are no more process crashes and there
are enough directed channels satisfying the ADD property, there is a
finite time from which a single process (namely, the correct process
$p_\ell$ with the smallest identity) sends messages \ALIVE$(\ell,n-1)$
and no other process $p_j$ sends the generating message
\ALIVE$(j,n-1)$.

\vspace{-.6cm}
\subsubsection{Message reception} (Lines~\ref{ADD-10}-\ref{ADD-17})
When a process $p_i$ such that $leader_i\neq i$ receives a message \ALIVE$(\ell,n-k)$,  it  learns that (a) $p_\ell$ is candidate
to be leader, and (b) there is a path with $k$ hops from  $p_j$
to itself.

If $\ell \leq leader_i$, $p_i$ considers $\ell$ as its current leader
(line~\ref{ADD-11}). Hence, if $\ell <leader_i$, $p_\ell$ becomes its new
leader, otherwise it discards the message. This is due to the fact
that $p_i$ currently considers $leader_i$ as leader, and the eventual
leader must be the correct process with the smallest identity.

Then, as the message \ALIVE$(\ell,n-k)$ indirectly comes from
$leader_i=\ell$ (which generated \ALIVE$(\ell,n-1)$) through a path
made up of $k$ different processes, $p_i$ increases the
associated timeout value if the timer $\timer_i[leader_i,hb]$ expired
before it received the message \ALIVE$(\ell,hb)$
(line~\ref{ADD-13}). Moreover, whether $\timer_i[leader_i,hb]$ expired
or not, $p_i$ resets $\timer_i[leader_i,hb]$ (line~\ref{ADD-14})
and starts a new monitoring session with respect to its current
leader and the cycle-free paths of length $hb$ from $leader_i$ to it.

The role of the timer $\timer_i[\ell,hb]$ is to allow $p_i$ to monitor
$p_{\ell}$ with respect to the forwarding of the messages
\ALIVE$(\ell,hb)$ it receives such that $hb=n-k$ (i.e., with respect
to the messages received from $p_j$ along paths of length $k$).

Finally, $p_i$ updates $hopbound_i[leader_i]$. To update it,
 $p_i$ first computes the value of
$not\_expired_i$ (line~\ref{ADD-15}) which is a bag of cycle-free path
lengths $x$ such that $timer_i[leader_i,n-x]$ is still
running\footnote{A bag (also called multiset or pool) is a ``set'' in
  which the same element can appear several times.  As an example,
  while $\{a,b,c\}$ and $\{a,b,c,b,b,c\}$ are the same set, they are
  different bags.}.
To this end, the idea then is to select the less penalized path (hence the
``smallest non-negative value'' at line~\ref{ADD-16}).  But, it is
possible that there are different cycle-free paths of lengths $x1$ and
$x2$ such that we have
$penalty_i[leader_i,n-x1]=penalty_i[leader_i,n-x2]$. In this case, in
a conservative way, $\mmax(n-x1,n-x2)$ is selected to update the local
variable $hopbound_i[leader_i]$. 

\vspace{-.5cm}
\subsubsection{Timer expiration} (Lines~\ref{ADD-18}-\ref{ADD-23})
Given a process $p_i$, when the timer currently monitoring its current
leader through a path of length $k=n-hb$ expires (line~\ref{ADD-18}),
it increases its $penalty_i[leader_i,n-k]$ entry (line~\ref{ADD-19}).

The entry $penalty_i[j,n-k]$ is used by $p_i$ to cope with the negative
effects of the channels which are on cycle-free paths of length $k$
from $p_j$ to $p_i$ and do not satisfy the ADD property. More
precisely we have the following. If, while $p_i$ considers $p_j$ is
its current leader (we have then $leader_i=j$), and $timer_i[j,n-k]$ expires,
$p_i$ increases $penalty_i[j,n-k]$. The values in the vector
$penalty_i[j,1..n]$ are then used at lines~\ref{ADD-15}-\ref{ADD-16}
(and line~\ref{ADD-22}) to update $hopbound_i[leader_i]$ which (if
$p_j$ is the eventually elected leader) will contain the length of
an cycle-free path 
 from $p_j$ to $p_i$ made up of $\Diamond$ADD channels (i.e., a
path on which $timer_i[j,n-k]$ will no longer expire).

Then, if for all the hopbound values, the timers currently monitoring
the current leader have expired (line~\ref{ADD-20}), 
$p_i$ becomes candidate to be leader (line~\ref{ADD-21}).

If one (or more) timer monitoring
its current leader has not expired, $p_i$ 
recomputes the path associated with the less penalized hopbound value
in order to continue monitoring $leader_i$ (line~\ref{ADD-22}).

%-------------------------------------------------------------------

%=========================================================================

\begin{algorithm*}[t!]
\centering{\fbox{
\begin{minipage}[t]{150mm}
\footnotesize 
\renewcommand{\baselinestretch}{2.5}
\resetline
\begin{tabbing}
aaaaa\=aaa\=aaa\=aaaa\=aaaaaa\=\kill

{\bf initialization}
$~~~~~~~~~~~~~~~~~~~~~~~~~~~~~~~~~~~~~~~~~~$ ----Code for $p_i$----\\

\line{ADD-01} \> $leader_i \leftarrow i$; $hopbound_i[i] \leftarrow n$;
         ${\set} ~\timer_i[i,n]~{\sf to}~ +\infty$; \\

\line{ADD-02}
\> {\bf fo}\={\bf r each} \= $j \in \{1,\cdots,n\} \setminus \{i\}$
     {\bf and each} $x\in \{1,\cdots,n\}$ {\bf do}

\\

\line{ADD-03} \>\> 
$\timeout_i[j,x] \leftarrow 1$;
${\set} ~\timer_i[j,x]~{\tto}~ \timeout_i[j,x]$;\\

\line{ADD-04} 
\>\> ${\set} ~\penal_i[j,x]~{\tto}~ -1$;        
         $hopbound_i[j] \leftarrow 0$ \\
  
\line{ADD-05} 
\> {\bf end for}.\\~\\

%----------------------------------------------------------
\line{ADD-06}
\> {\bf every $T$ time units of} $\clock_i()$ {\bf do}\\

\line{ADD-07}
\> \> {\bf if} \= 
 $(hopbound_i[leader_i]>1)$  {\bf then} \\

\line{ADD-08}
\>\>\> {\bf for each} $j\in out\_neighbors_i$ {\bf do} 

  $\send$ \ALIVE$(leader_i,hopbound_i[leader_i]-1)~\tto~ p_j$  {\bf end for}\\

\line{ADD-09}
\> \> {\bf end if}.\\~\\

%----------------------------------------------------------
\line{ADD-10}
\> {\bf when \ALIVE$(\ell,hb)$ such that $\ell\neq i$
      is received} \\
\>      ~~\% from a process in $in\_neighbors_i$ \\

\line{ADD-11}
\>\> {\bf if} \= $(\ell\leq leader_i)$ {\bf then} \\

\line{ADD-12}
\>\>\>  \=$leader_i\leftarrow \ell$; \\

\line{ADD-13}
\>\>\>\> 
    {\bf if} $([timer_i[leader_i,hb]$ expired$) $~$ {\bf then}$    $\timeout_i[leader_i,hb] \leftarrow \timeout_i[leader_i,hb] \times 2 $ {\bf end if};\\

\line{ADD-14}
\>\>\>\>
 $\set~ timer_i[leader_i,hb]~\tto~ \timeout_i[leader_i,hb]$;\\
 
\line{ADD-15}
\>\>\> \>
$not\_expired_i \leftarrow \{ x~| ~timer_i[leader_i,x]$ not expired~$\}$;\\

\line{ADD-16}
\>\>\>\>
$hopbound_i[leader_i] \leftarrow \mmax\{x\in not\_expired$
 with smallest non-negative $penalty_i[leader_i,x]\}$ \\
  
\line{ADD-17} \>\> {\bf end if}.\\~\\

\line{ADD-18} 
\> {\bf when $timer_i[leader_i,hb]$ expires} and $(leader_i\neq i)$ {\bf do} \\

\line{ADD-19} 
\>\>\> $penalty_i[leader_i,hb] \leftarrow penalty_i[leader_i,hb] + 1$;\\

\line{ADD-20}
\>\>\>
{\bf if} $\big(\wedge_{1\leq x\leq n} ([timer_i[leader_i,x]$ expired$)\big)$ {\bf then}\\
   
\line{ADD-21} 
\>\>\>\>\>  \= $leader_i \leftarrow i$\\

\line{ADD-22} 
\>\>\>\>
  {\bf else} \> same as lines~\ref{ADD-15}-\ref{ADD-16}\\

\line{ADD-23}
\>\>\> 	{\bf end if}.

\end{tabbing}
\end{minipage}
}
\vspace{0.1cm}
\caption{Eventual leader election in the $\Diamond$ADD model
  with known membership}
\label{algo:description}
}
\end{algorithm*}

%========================================================================
\vspace{-.5cm}

\section{Proof of Algorithm~\ref{algo:description}}
\label{sec:proof}

This section shows that Algorithm~\ref{algo:description} elects an eventual leader while assuming the Span-Tree behavioral assumption.

%KVbegin
We have to prove that the algorithm satisfies \emph{Validity} and
\emph{Eventual Leader Election}.  For \emph{Validity}, let us observe
that the local variables $leader_i$ of all the processes always
contain a process identity. Hence, we must only prove \emph{Eventual
  Leader Election}, i.e. we must only show that the variables
$leader_i$ of all the correct processes eventually converge to the
same process identity, which is the identity of one of them.

Due to space limitation,  the proof of the lemmas are in the Appendix. 

\vspace{-.2cm}
\begin{lemma}
Let $p_i$ and $p_j$ be two correct processes connected by a
$\Diamond${\em ADD} channel, from $p_i$ to $p_j$. There is a time
after which any two consecutive messages received by $p_j$ on this
channel are separated by at most $\Delta =(K-1)\times T+D$ time units.
\label{lemma:bounded-delay}
\end{lemma} 

\vspace{-.2cm}

Given any run $r$ of Algorithm~\ref{algo:description}, let $\correct(r)$
denote the set of processes that are correct in this run and $\crashed(r)$
denote the set of processes that are faulty in this run.

%KVbegin
The following lemma shows that there is a time after which there are no \ALIVE$(j,n-k)$ messages with $p_j \in \crashed(r)$, i.e. eventually all correct processes stop sending the \ALIVE \ messages from a failed process which proves that once a leader fails, eventually all processes elect a new leader.
%KVend
%-------------------------------------------------------------------------
\vspace{-.2cm}
\begin{lemma}
\label{onlyCorrect}
Given a run $r$, there is a time $t^a$ after which there are no 
messages \ALIVE$(i,n-a)$ with $p_i \in \crashed(r)$ and $1 \leq a < n-1$. 
\end{lemma} 
\vspace{-.2cm}
%---------------------------------------------------------------------
\vspace{-.2cm}

\begin{theorem}
\label{main-theorem}
 Given a run $r$ satisfying the {\em Span-Tree} property, there is a
 finite time after which the variables $leader_i$ of all the correct
 processes contain the smallest identity $\ell\in \correct(r)$.
 Moreover, after $p_\ell$ has been elected, there is a finite time after 
 which the only messages sent by processes are \ALIVE$(\ell,-)$ messages.
\end{theorem}

\vspace{-.2cm}
\vspace{-.2cm}

%---------------------------------------------------------------------
\begin{theorem}
\label{cost-theorem}
The size of a message is $O(\llog~n)$. 
\end{theorem}

\begin{proof}
The proof follows directly from the fact that a message carries a
process identity which belongs to the set $\{1,\cdots,n\}$ and a hopbound
number $hopbound$ such that $2\leq hopbound \leq n-1$. Since an integer bounded with $n$ can be represented with exactly $\llog~n$ bits and we have two integers bounded with $n$ we have that the size of every message is $O(\llog~n)$.
\end{proof}
\vspace{-.2cm}

%==========================================================================================
\vspace{-.5cm}
\subsection{Time complexity}
%\vspace{-.1cm}
%In this section the time complexity of Algorithm 1.  
Given a run $r$, let $\ell$ denote the smallest identity such that
$\ell \in correct(r)$. Let $t^a$ be the time given by
Lemma~\ref{onlyCorrect}, i.e. a time from which no message from
crashed processes is till in transit (they have been received or are
lost).  Let $t^a \leq t^r$ be the time after which:
\begin{enumerate}
 \item All failures already happened.
 \item All $\Diamond$ ADD channels satisfy  their constants $K$ and $D$.
\end{enumerate}

%KVbegin
After $t^r$, let $\Delta$ be the constant given
by Lemma~\ref{lemma:bounded-delay}.
%KVend

\begin{lemma}
\label{lemmaUpperBound}
Let $p_i$ be a correct process such that for every $t > t^r$,
$hopbound_i[\ell] = n-k$. Then, for every correct process $p_j$ such
there is a  $\Diamond$ ADD channel from $p_i$ to $p_j$,
$timeout_j[\ell,n-(k+1)] \leq \mathcal{C} + 2^{log(\lceil \Delta\rceil
  )}$ with $timeout_j[\ell,n-(k+1)]= \mathcal{C}$ before $t^r$.
\end{lemma}

%----------------------------------------------------------------------------------------------------------------------------

Lemma~\ref{lemmaUpperBound} states that after $t^r$, a timeout value
is increased a finite number of times. Let $t^c$ be the time after
which all timeouts have reached their maximum, namely, no timeout is
increased again.
%SR:check
The following claims refer to the communication graph after $t^c$.

%---------------------------------------------------------------------------------------------------------------------------
\begin{lemma}
\label{lemma:timedelta}
 For every correct process $p_i$ such that there is a minimum length
 path of $\Diamond$~ADD channels of length $k$ from $p_\ell$ to $p_i$,
 $leader_i = \ell$ at time $t^c + (k \times \Delta)$.
\end{lemma}
%---------------------------------------------------------------------------------------------------------------------------

Let $\mathcal{D}$ be the diameter of the underlying spanning-tree of
$\Diamond$~ADD channels.
\begin{theorem}
\label{th:time}
 For every correct process $p_i$ it takes $O(\mathcal{D} \cdot
 \Delta)$ time to have $leader_i = \ell$.
\end{theorem}

\begin{proof}
The proof is direct from Lemma~\ref{lemma:timedelta}.
\end{proof}

\vspace{-.5cm}

\subsection{Simulation experiments}
\label{sec:simExp}

\begin{wrapfigure}[42]{R}{16em}
\includegraphics[scale=.35]{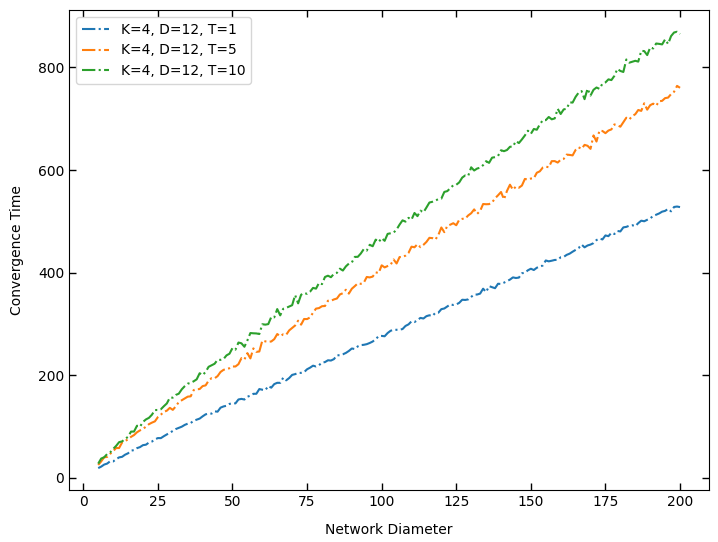}
\caption{A ring with drop rate of $1\%$} 
\label{fig:cycle-diam1to200-rate1}

 \includegraphics[scale=.35]{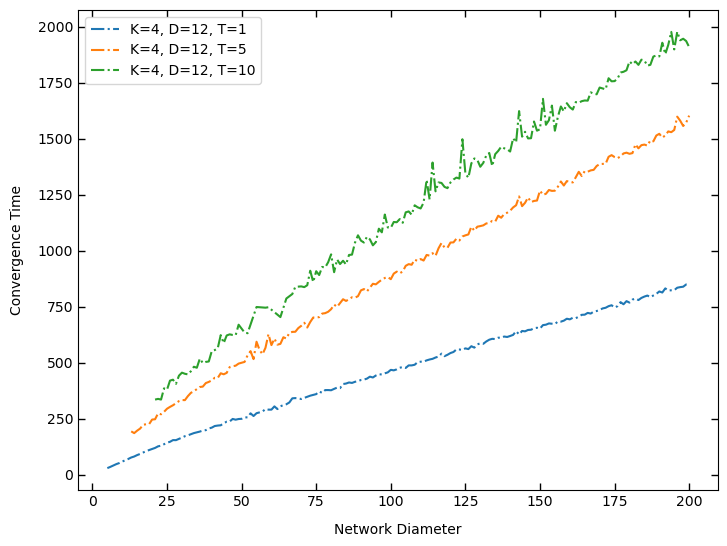}
 \caption{A ring with drop rate of $99\%$} 
 \label{cycle-diam1to200-rate99}

 \includegraphics[scale=.5]{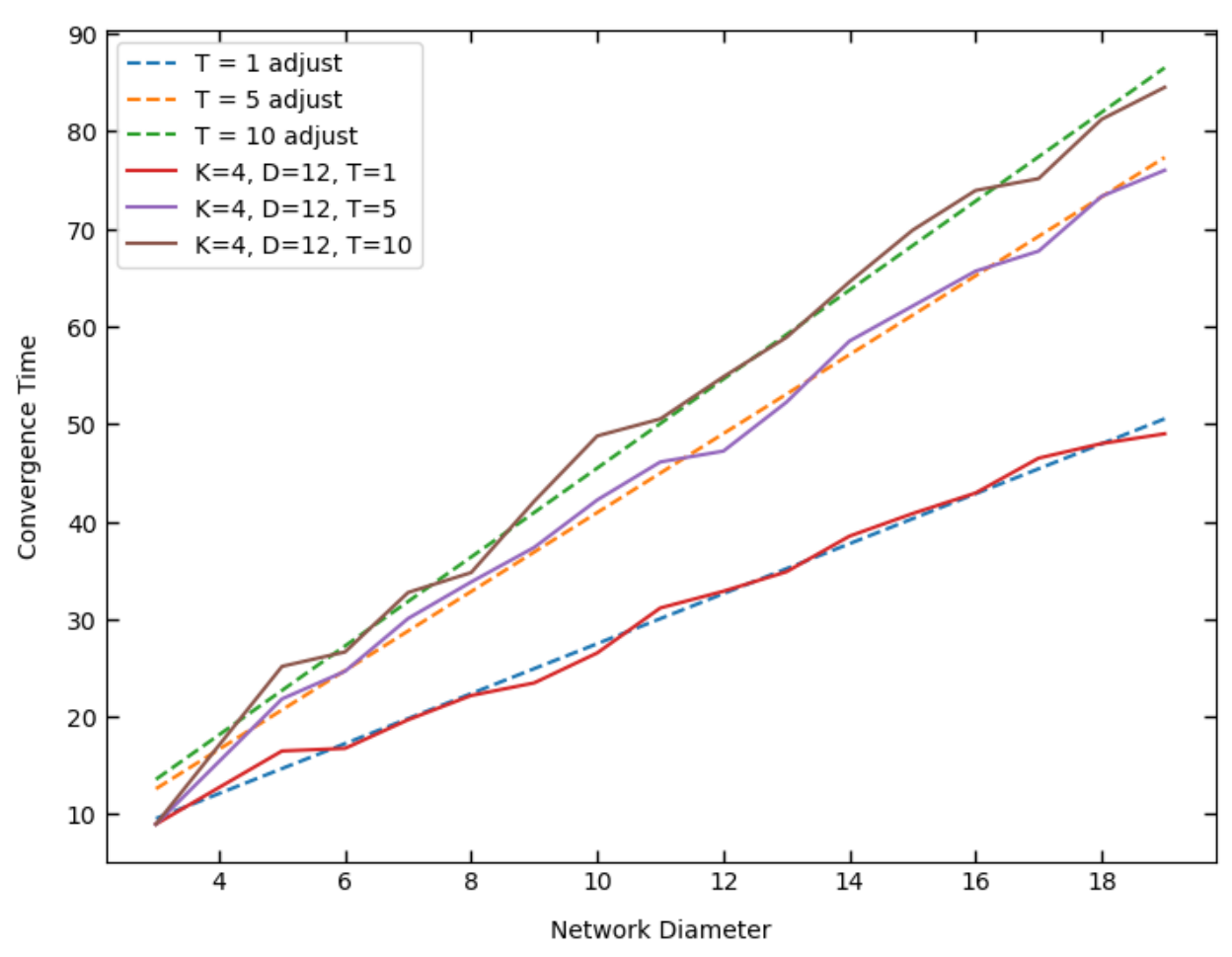}
 \caption{A $3$-regular random graph with drop rate of $1\%$} 
 \label{fig:random-diam1to20-rate10}

\end{wrapfigure}
This section presents simulation experiments related to the
performance predicted by Theorem~\ref{th:time} of
Algorithm~\ref{algo:description}.  Only a few experiments are
presented, a more detailed experimental study is beyond the scope of
this conference version. Our experiments show that a leader is elected
in time proportional to the diameter of the network, in two network
topologies: a ring and a random regular graph of degree $3$.
%SR: not clear how the experimental results are related to this Lemma
%We want to experimentally show the results given in Theorem~\ref{th:time}.

Considering the constants $K$ and $D$ satisfied by an $\Diamond$~ADD
once it stabilizes, Lemma~\ref{lemma:bounded-delay} shows that for a
given $T$ (the frequency with which the messages are sent), then
$\Delta =(K-1)\times T+D$ is an upper bound on the time of the
consecutive reception of two messages by a process.  According to
Theorem~\ref{th:time}, the time to elect a leader is proportional to
the diameter of the network, where the $K$, $D$ and $T$ determine the
slope of the function.

For the (time and memory) efficiency of the experiments we assume some
simplifying assumptions, which seem sufficient to a preliminary
illustration of the results:
\begin{itemize}
 \item All the channels are $\Diamond$~ADD to avoid the need of a
   penalization array.
 \item All the messages are delivered within time at most $D$ or not
   delivered at all.  This is sufficient to illustrate the convergence
   time to a leader. Additional experimental work is needed to
   determine the damage done by messages that are delivered very late.
 \item We selected $K=4$, $D=12$ and $T=1,5,10$.
 \end{itemize}

\vspace{-.7cm}
 
\subsubsection{Convergence experiments} 
The experiments of the ring in
Figure~\ref{fig:cycle-diam1to200-rate1} and Figure~\ref{cycle-diam1to200-rate99}, are when the probability of a message being lost is $1\%$,
and $99\%$ respectively.
The case of a random graph of degree $3$ up to 50,000 nodes is in
Figure~\ref{fig:random-diam1to20-rate10} when the
probability of a message being lost is $1\%$. These experiments
verify that indeed the convergence time is proportional to the
diameter.  The constants appear to be smaller than $\Delta$, the one
predicted by Theorem~\ref{th:time}.

\vspace{-.3cm}

\subsubsection{Simulation details}
We performed our simulation results in a 48 multicore machine with
256GB of memory, using a program based on the Discrete Event Simulator
\emph{Simpy}, a framework for Python. We used the  \textit{Networkx} package to
model graph composed of ADD channels. For the ring simulations, experiments were performed
for each $n$ from 10 up to 400 nodes, and taking the average of 10
executions, for each value of $n$.  For the random regular networks,
the degree selected was $3$, and experiments starting with $n$
starting in $100$, up to $10,000$, taking the average of $5$
executions. The $n$ was incremented by 100 to reach 10,000 and from
then on until 50,000 we incremented $n$ by 10,000 each time. A
performance impediment was indeed the large amount of memory used.

The convergence  time curves we obtained for the ring experiment are
functions of the form $f(x) =c\cdot x$, where $x$ represents the
diameter of the network, and the constant $c$ is, roughly, between
$2.5$ and $4.5$ as $T$ goes from $1$ to $10$. While for the random
regular networks, we again got a constant that doubled in size,
roughly, as $T$ goes from $1$ to $10$.  This behavior seems to be
better than the one predicted by Theorem~\ref{th:time}, which says
that the constant $c$ should have grown 10 times.

\subsubsection{Re-election convergence simulation}
\vspace{-.3cm}

If an elected leader fails, we would like to know in how much time a new leader is elected.

Note that the $\Diamond$~ADD channels can arbitrarily delay the
delivery of some messages. This condition has a great impact in the
time it takes to Algorithm~\ref{algo:description} to change a failed
leader. For the following simulations again we assume that all the
messages are delivered within time at most $D$ or not delivered at
all. But note that in a realistic scenario, we can ease the impact of
the arbitrarily delayed messages by adding a timestamp to every message and keeping
track for every neighbor  of this timestamp. If the timestamp of the
recently received message is smaller than the current one, just ignore
the message. This timestamp does not have a bound, but if we use an
integer and increase it by one every second that a message is sent,
this integer can hold on up for a century without
overflowing \footnote{An unsigned integer can be encoded with 32
  bits, so its maximum value can be $4294967296$. A year has
  $31536000$ seconds.}. By adding an integer to the message, we keep
messages of size $O(log \ n)$.
%KV I think that using this technique is almost the same that not assuming late messages at all 

\setlength\intextsep{0pt}
 \begin{wrapfigure}[14]{I}{0.5\textwidth}
%\begin{figure}[t]
 
\centering
\includegraphics[width=0.5\textwidth]{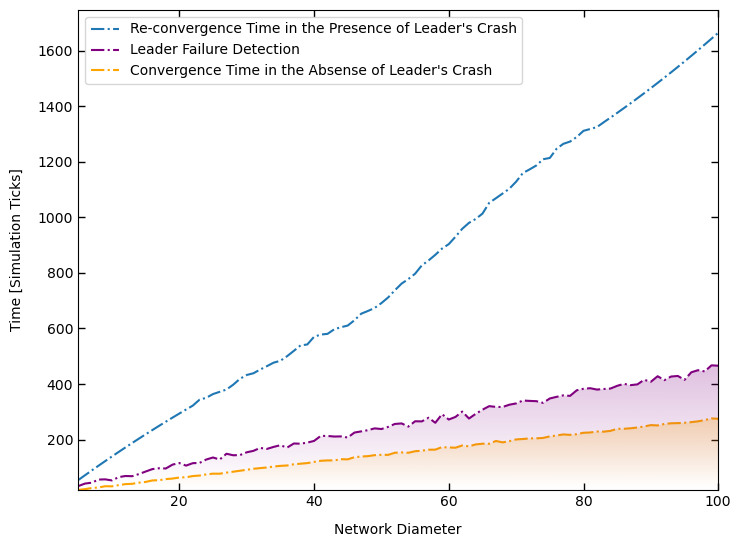}
\caption{Convergence time for re-election} 
\label{fig:cycle-diam-recovery}
%\end{figure}

 \end{wrapfigure}
 
For the simulation of Figure~\ref{fig:cycle-diam-recovery} we selected $K=4$, $D=12$, $T=1$ and the probability of a message being lost is $1\%$ . We performed this simulation on a ring. The algorithm starts at time $t_0$ and continues its execution till the average time in which a leader is elected (the curve represented in orange). In this time, the candidate to be the leader fails and then a timer from an external observer is started in every process. This timer is used to know the average time needed for each  process to discard the failed leader (curve represented in purple)  and then converge to a new leader (curve represented in blue).  This experiment
verify that indeed the convergence time after the current leader fails is proportional to the
diameter since $\Delta = (K-1)\times T+D = 3 + 12 = 15$.

% \paragraph{}
% \vspace*{-\parskip}

\vspace{-.5cm}
%========================================================================
\section{Eventual Leader Election with Unknown Membership}
\label{sec:unknown-membership}

\vspace{-.2cm}

%SR
Here, while $n$ exists and has a fixed value, it is no longer assumed
that processes know it. Consequently, the  processes have an
 ``Unknown Membership'' of how many and which are the processes in the network.
Nevertheless,  for convenience,
the proposed algorithm still uses the array notation for
storing  the values of timers, timeouts, hopbounds, etc.
(in an implementation dynamic data structures --e.g., lists--
should be used). 
 
Algorithm~\ref{algoUNK:description} solves eventual leader election in
the $\Diamond$~ADD model with unknown membership, which means that,
initially, a process knows nothing about the network, it knows only
its input/output channels.  

%KVbegin
Our goal is to maintain the $O(\llog~ n)$ bound on the size of the
messages even in this model.  It seems that it is not easy to come up
with a minor modification of the first algorithm. For instance, a
classic way of ensuring that forwarding the \ALIVE \ message is
cycle-free is to include the path information in the message along which
the forwarding occurred, as done in the paper~\cite{KW19}. This would
result in message sizes of exponential size, while assuming a slightly
different model, we show how to eventually stay with $O(\llog ~ n)$
messages.

Furthermore, since we want the complexity to be $O(\llog ~n)$
eventually, we need to design a mechanism that works as a
\emph{broadcast} in which once a process $p_i$ knows a new process
name from $p_j$, the later does not need to send to $p_i$ the same
information but only the leader information. The proposed mechanism in
this paper is not the same as the proposed in~\cite{VR19} since we are
preventing processes to send all the known names but eventually, only
the leader information.

Since no process has knowledge about the number of participating
processes, this number must be learned dynamically as the names of
processes arrives. In order to the leader to reach every process in
the network , there must be a path of $\Diamond$~ADD channels from
every correct process to the leader. It follows that an algorithm for
eventual leader election in networks with unknown membership cannot be
a straightforward extension of the previous algorithm. More precisely,
instead of the unidirectional channels and {\it Span-Tree}
assumptions, Algorithm~\ref{algoUNK:description} assumes that (i) all
the channels are bidirectional $\Diamond$~ADD channels, and (ii) the
communication network restricted to the correct processes remains
always connected (namely, there is always a path --including correct
processes only-- connecting any two correct processes).

%KVend

In Algorithm~\ref{algo:description}, every process $p_i$ uses $n$ to
initialize its local variable $hopbound_i[i]$ (which thereafter is
never modified).  In the unknown membership model,
$hopbound_i[i]$ is used differently, namely it represents the number of
processes known by $p_i$ so far. So its initial value is $1$. 
Then, using a technique presented in~\cite{VR19}, $hopbound_i[i]$ is updated 
 as processes know about each
other: every time a process $p_i$ discovers a new process identity
 it increases $hopbound_i[i]$.

\vspace{-.3cm}

 %=========================================================================
\subsection{General principle  of the algorithm}

\vspace{-.1cm}

Initially each process $p_i$ only knows itself and how many channels
are connected to it. So the first thing $p_i$ needs to do is
communicate its identity to its neighbors. Once its neighbors know
about it, $p_i$ no longer sends its identity. The same is done with other
names that $p_i$ learns. For that, $p_i$ keeps a \emph{pending set}
for every channel connected to it that tracks the information it needs
to send to its neighbors. So initially, $p_i$ adds the pair $(\new,i)$
to every pending set.

During a finite amount of time, it is necessary to send an \ALIVE$()$
message to every neighbor without any constraint because the set of
process names needs to be communicated to other processes. That is, 
information about a leader might be empty and the message only
contains the corresponding pending set.

When process $p_i$ receives an \ALIVE$()$ message from $p_j$, this
message can contain information about the leader and the corresponding
pending set that $p_j$ saves for $p_i$. First, $p_i$ processes the
information contained in the pending set and then processes the
information about the leader.
%\vspace{-.5cm}

\vspace{-.5cm}

\subsubsection{How $p_i$ learns new process names}
If $p_i$ finds a pair with a name labeled as new and is not aware of it,
it stores the new name in the set $known_i$, increases its hopbound
value, and adds to every pending set (except to the one belonging to
$p_j$) this information labeled as new. In any case, $p_i$ needs to
communicate $p_j$ that it already knows that information, so $p_i$
adds this information to the pending set of $p_j$ but labeled as an
acknowledgment.

When $p_j$ receives $name$ labeled as an acknowledgment from $p_i$,
i.e. $(\ack,name)$, it stops sending the pair $(\new,name)$ to it, so
it deletes that pair from $p_i$'s pending set. Eventually, it
receives a pending set from $p_j$ not including $(\new,name)$, so
$p_i$ deletes $(\ack,name)$ from $p_j$'s pending set.

\vspace{-.5cm}

\subsubsection{How $p_i$ processes the leader information}
As in Algorithm~\ref{algo:description}, every process keeps as leader
a process with minimum id. Since it is assumed that all the channels
are $\Diamond$ ADD, there is no need to keep a timer for every
hopbound value or a penalty array. In this case, process $p_i$ keeps
the greatest $n-k$, i.e. hopbound value that it receives from the
process it considers to be the leader. If this value (or a greater
one) does not arrive on time, $p_i$ proposes itself as the leader. In
case a smaller hopbound value of the leader arrives, it is only taken
if its timer expired.

\vspace{-.3cm}

% %---------------------------------------------------------------------
 \subsection{Local variables at each process $p_i$}
 
 \vspace{-.2cm}

  Each process $p_i$ manages the following local variables.
   \vspace{-.2cm}

 \begin{itemize}
 \item $leader_i$  contains the identity of the candidate leader. 
 \item $hopbound_i[1..)$ is an array of natural numbers;
   $hopbound_i[i]$ is initialized to $1$.
 \item $timeout_i[\cdot]$ and $timer_i[\cdot]$ have the same meaning as in
   Algorithm~\ref{algo:description}. So, when $p_i$ knows $p_j$,
   the pair $\langle \timer_i[j],\timeout_i[j]\rangle$ is used by $p_i$
   to monitor the sending of messages by $p_j$ (which is not
   necessarily a neighbor of $p_i$).
 \item $known_i$ is a set  containing the  processes  currently
  known by  $p_i$. At the  beginning, $p_i$ only knows itself.
 \item $out\_neighbors_i$ is a  set  containing the names of the
   channels connecting $p_i$ to its neighbor processes.
   The first time   $p_i$ receives through channel $m$  a
   message sent by a process $p_j$, $p_j$ and $m$ become
   synonyms 
   %from a neighbor  addressing point of view. 
 \item $pending_i[1,...,k]$ is a new  array in which, when $p_i$ knows $p_j$,
   $pending_i[j]$
   contains the pairs of the form $(label, id)$ that are pending to
   be send through channel connecting $p_i$ and $p_j$. 
   There are two possible  labels, denoted  $\new$ and $\ack$. 
 \end{itemize}
%---------------------------------------------------------------------
\begin{algorithm*}[h!]
\centering{\fbox{
\begin{minipage}[t]{150mm}
\footnotesize 
\renewcommand{\baselinestretch}{2.5}
\resetline
\begin{tabbing}
aaaaa\=aaa\=aaa\=aaaa\=aaaaaa\=\kill

\line{RELS-01} 
\> {\bf initialization}
$~~~~~~~~~~~~~~~~~~~~~~~~~~~~~~~
~~~~~~~~~~~~~~~~~~~~~~~~~~~~~$ ----Code for $p_i$----\\

\line{RELS-02} \> $leader_i \leftarrow i$; $hopbound_i[i] \leftarrow 1$; \\

\line{RELS-03} \> $known_i \leftarrow \{i\}$;
$out\_neighbors_i$ initialized to the channels of $p_i$;\\

\line{RELS-04} \> {\bf for  each}  {$m \in out\_neighbors_i$} {\bf do}
           $pending_i[m] \leftarrow \{(\new,i)\}$ {\bf end for}.   \\~\\

%----------------------------------------------------------

\line{RELS-05}
\> {\bf every $T$ time units of} $\clock_i()$ {\bf do}\\
\line{RELS-06}
\>\> {\bf for each channel}  $m \in out\_neighbors_i$
  (let $p_j$ be the associated neighbor)  {\bf do} \\
\line{RELS-07}
\>\> \> {\bf if} \=  $(hopbound_i[leader_i]>1)$ \\

\line{RELS-08}\>\>\>\> 
         {\bf then} \= $\send$
         \ALIVE$(leader_i,hopbound_i[leader_i]-1, pending_i[j])~\tto~  p_j$ \\
\line{RELS-09}
\>\> \> \>  {\bf else}\>
  $\send$ \ALIVE$(\bot,\bot, pending_i[j])~\tto~  p_j$  \\

\line{RELS-10}
\>\> \> {\bf end if}\\

\line{RELS-11}
\>\> {\bf end for}.\\~\\

%----------------------------------------------------------

\line{RELS-12}
\> {\bf when \ALIVE$(\ell,hb,pending)$ is received from}
$p_j$  {\bf through channel} $m$ \\
\>\% from then on: $p_j$ and $m$ are synonyms from an addressing point of view\\
   
\line{RELS-13}
\>\> $set_i \leftarrow \emptyset$;    \\
    
\line{RELS-14}
\>\> {\bf for each} $(label,k) \in pending$ {\bf do}\\

\line{RELS-15}
\>\>\> {\bf if} \=  $(label = \new)$ \\

\line{RELS-16}
\>\>\>\> {\bf then} \> $set_i \leftarrow set_i \cup \{k\}$;\\

\line{RELS-17}
\>\>\>\> \>  {\bf if} \= ($k \notin known_i$)\\

\line{RELS-18}
\>\>\>\>\> \> {\bf then} \=
           $known_i \leftarrow known_i \cup \{k\}$;
           $hopbound_i[i] \leftarrow hopbound_i[i]+1$; \\

\line{RELS-19}
\>\>\>\>\>\> \>add an entry in $timeout_i$, $timer$, $hopbound_i$; \\

\line{RELS-20}
\>\>\>\>\> \>\> add $(\new,k)$ to every $pending[p]$ with $p \neq m$\\
 
\line{RELS-21}
\>\>\>\>\> {\bf else} \= {\bf if} \= $((\new,k) \in pending_i[m])$ \\

\line{RELS-22}
\>\>\>\>\>\>$~~~~$
    {\bf then} \=
    $pending_i[m] \leftarrow pending_i[m] \setminus  (\new,k)$ {\bf end if}; \\

\line{RELS-23}
\>\>\>\>\>\>  $pending_i[m] \leftarrow pending_i[m] \cup (\ack,k)$\\

\line{RELS-24}
\>\>\>\> \> {\bf end if} \\

\line{RELS-25}
\>\>\>\> {\bf else} $pending_i[m] \leftarrow pending_i[m] \setminus (\new,k)$\\

\line{RELS-26}
\>\>\>  {\bf end if}\\

\line{RELS-27}
\>\> {\bf end for}; \\

\line{RELS-28}
\>\> {\bf for each} $(\ack,k) \in pending_i[m] $ such that $k \notin set_i$
      {\bf do} \\

\line{RELS-29}
\>\>\> $~~~~$ 
$pending_i[m] \leftarrow pending_i[m] \setminus \{(\ack,k)\}$ {\bf end for};\\

\line{RELS-30}
\>\> {\bf if} \= $(\ell\leq leader_i$ and $\ell \neq i)$ \\

\line{RELS-31}
\>\>\> {\bf then} \=$leader_i\leftarrow \ell$; \\

\line{RELS-32}
\>\>\> \> {\bf if} \=
  $(hb\geq hopbound_i[leader_i])~\vee~(\timer_i[leader_i]$ expired$)$ \\

\line{RELS-33}
\> \> \>  \> \> {\bf then} \= $hopbound_i[leader_i]\leftarrow hb$;\\

\line{RELS-34}
\>\>\>\>\>\>
  {\bf if} $([timer_i[leader_i]$ expired$)$ \\

 \line{RELS-35}
    \>\>\>\>\>\> $~~~~$
    {\bf then} $\timeout_i[leader_i,hb] \leftarrow \timeout_i[leader_i,hb] \times 2 $
  {\bf end if};\\

\line{RELS-36}
\>\>\> \>\>\>
  $\set~ timer_i[leader_i]~\tto~ \timeout_i[leader_i]$\\

\line{RELS-37} \>\>\>\>  {\bf end if}\\

\line{RELS-38}  \>\>  {\bf end if}.\\~\\

\line{RELS-39}
\> {\bf when} $(timer_i[leader_i]$ expires) {\bf do} $leader_i \leftarrow i$. 
      
\end{tabbing}
\end{minipage}
  }
 \vspace{0.1cm} 
 \caption{Eventual leader election  in the  $\Diamond$  ADD model
   with unknown membership}
\label{algoUNK:description}
}
\end{algorithm*}
%----------------------------------------------------------------------

%-----------------------------------------------------------------

\vspace{-.7cm}

\subsection{Detailed behavior of a process $p_i$}
%\vspace{-.2cm}

The code of Algorithm~\ref{algoUNK:description} addresses two
complementary issues: the management of the initially unknown
membership, and the leader election.

\vspace{-.6cm}
\subsubsection{Initialization}(Lines~\ref{RELS-01}-\ref{RELS-04})
Initially, each process $p_i$  knows only  itself and how many 
input/output channels it has. Moreover, it
does not know the name of the processes connected to these channels (if any) 
and how many neighbors it has (the number of channels is higher or equal to
the number of neighbors).  So when the
algorithm begins, it proposes itself as the leader and in the pending
sets of every channel adds its pair $(\new,i)$ for neighbors to know
it.

\vspace{-.5cm}
\subsubsection{Sending a message}(Lines~\ref{RELS-05}-\ref{RELS-11})
Every $T$ units of time, $p_i$ sends a message through every channel
$m$. In some cases the leader information is empty because of the
condition of line~\ref{RELS-07}. But in any case, it must send a
message that includes information about the network that is included
in the set $pending_i[j]$.

\vspace{-.5cm}
\subsubsection{Receiving a message}(Lines~\ref{RELS-12}-\ref{RELS-38})
When $p_i$ receives a message (line~\ref{RELS-12}) from process $p_j$
(through channel $m$), at the beginning it knows from which
channel it came and eventually knows from whom is from. When the
message is received,    the information included in
$pending$ (lines~\ref{RELS-14}-\ref{RELS-29}) is processed, and then
the leader information is processed (lines~\ref{RELS-30}-\ref{RELS-38}).

\vspace{-.5cm}
\subsubsection{Processing  new information}(Lines~\ref{RELS-14}-\ref{RELS-29})
The input parameter set $pending$ includes pairs of the form $(label,
id)$, where $label \in \{\new,\ack\}$ and $id$ is the name of some
process. When $p_i$ processes the pairs that it received from $p_j$
there can be two kind of pairs. The first  is a pair with label
$\new$ (line~\ref{RELS-15}), which means that $p_j$ is sending new
information (at least for $p_j$) to $p_i$. When this information is
actually new for $p_i$ (line~\ref{RELS-17}) then, it stores this new
name, increases its hopbound entry and adds to every pending set(but
not the one from which it received the information) this new
information (line~\ref{RELS-20}).

In case that $p_i$ already knows the information labeled as new for
$p_j$ (line~\ref{RELS-21}), then $p_i$ needs to check if it is
included in the pending set to $p_j$ this information as new too. If
that is the case, then it deletes from $pending[m]$ this pair
(line~\ref{RELS-22}). In any case, $p_i$ adds to the pending set the
pair $(\ack,k)$ for sending through the channel from where this
message was received (line~\ref{RELS-23}).

If $p_i$ receives the pair $(\ack,k)$ (line~\ref{RELS-25}), then it
deletes the pair $(\new,k)$ from the set $pending_i[m]$, because the
process that sent this pair, already knows $k$.

\vspace{-.5cm}
\subsubsection{Processing the leader related  information}
(Lines~\ref{RELS-30}-\ref{RELS-38}).
If the leader related information is not empty, $p_i$
processes it. As in the first algorithm, if the identity of the proposed
leader is smaller than the current one, then it is set as $p_i$'s new
leader (line~\ref{RELS-31}). Then, it processes the hopbound. If the recently
arrived hopbound is greater than the one currently stored, then the recently
arrived is set as the new hopbound (line~\ref{RELS-33}). If the timer
for the expected leader expired, it needs more time to arrive to
$p_i$, so the timeout is increased (line~\ref{RELS-35}) and the timer
is set to timeout (line~\ref{RELS-36}).

\vspace{-.5cm}
\subsubsection{Deleting pairs}
(Lines ~\ref{RELS-21},~\ref{RELS-25} and ~\ref{RELS-28}).
If some process $p_i$ wants to send some information $k$ to $p_j$, it
adds to the pending set of $p_j$ the pair $(\new,k)$. When $p_j$
receives this pair, it looks if this is already in its set, in that
case, it deletes the pair from $p_i$'s pending set
(line~\ref{RELS-21}). Then, $p_j$ adds an $(\ack,k)$ to the pending set
of $p_i$. As soon as $p_i$ receives this pair from $p_j$, it deletes
from $p_j$'s pending set the pair $(\new,k)$ (line ~\ref{RELS-25}). So when $p_j$ receives a pending set from $p_i$ without the pair
$(\new,k)$, it means that $p_i$ already received the acknowledgment
message, so $p_j$ deletes $(\ack,k)$ from $p_i$'s pending set
(line~\ref{RELS-28}).

\vspace{-.5cm}
\subsubsection{Timer expiration}(Line~\ref{RELS-39}).
When the timer for the expected leader expires, $p_i$ proposes itself
as the leader.

Notice that, when compared to Algorithm~\ref{algo:description}, 
Algorithm~\ref{algoUNK:description} does not use 
the local arrays $penalty_i[1..n,1..n]$ employed to monitor the paths made of
non-ADD channels.

%=======================================================================
\vspace{-.3cm}

\section{Underlying  Behavioral Assumption and
  Proof of Algorithm~\ref{algoUNK:description}}
  \label{sec:proofUknown}
\vspace{-.2cm}

\subsubsection{Basic behavioral assumption}
In the following we consider that there is a time $\tau$ after which
no more failures occur, and the network is such that (i) all the
channels are bidirectional $\Diamond ADD$ channels, and (ii) the
communication network restricted to the correct processes remains
always connected. Assuming this, this section shows that
Algorithm~\ref{algoUNK:description} eventually elects a leader despite
initially unknown membership. All the proofs of this algorithm are in the appendix. 
%========================================================================

\vspace{-.5cm}

\section{Conclusion}
\label{sec:concl}
\vspace{-.2cm}

\label{sec:conclusion}
The $\Diamond$ADD model has been studied in the past as a realistic,  particularly weak 
communication model. A channel from a process $p$ to a 
process $q$ satisfies the $ADD$ property if there are two integers $K$
and $D$ (which are unknown to the processes) and a finite time $\tau$
(also unknown to the processes) such that, after $\tau$, in any
sequence of $K$ consecutive messages sent by $p$ to $q$ at least one
message is delivered by $q$ at most $D$ time units after it has been
sent.  Assuming first that the correct processes are connected by a
spanning tree made up of $\Diamond$~ADD channels, this article has
presented an algorithm that elects an eventual leader, using messages
of only size $O( ~\llog~ n)$. Previous algorithms in the $\Diamond$ADD
model implemented an eventually perfect failure detector, with
messages of size $O(n  ~\llog~ n)$.  In addition to this, the article has
presented a second eventual leader election algorithm in which no
process initially knows the number of processes.  This algorithm
 sends larger messages, to be able to estimate $n$,
but only for a finite amount of time, after which the size of the
messages is again $O( ~\llog~ n)$. We conjecture that it is necessary,
that the process identities are repeatedly communicated to the
potential leader.  Although we proved that our algorithms elect a leader in time proportional to the diameter
of the graph,
many interesting question related to performance
remain open.
 
%=========================================================================
% \section*{Acknowledgments}
% This work has been partially supported by PAPIIT UNAM and the French
% ANR project 16-CE40-0023-03 DESCARTES devoted to layered and modular
% structures in distributed computing.

%=======================================================================

%=======================================================================

\newpage
\pagenumbering{roman}

\appendix
\section{Proofs omitted from the main text}

\textbf{Lemma~\ref{lemma:bounded-delay}}\textit{
Let $p_i$ and $p_j$ be two correct processes connected by a
$\Diamond${\em ADD} channel, from $p_i$ to $p_j$. There is a time after which any two
consecutive messages received by $p_j$ on this channel
are separated by at most $\Delta =(K-1)\times T+D$ time units.}\\

\begin{proof}
For the channel $(p_i,p_j)$, let us consider a time
from which it satisfies the ADD property. 
Due to the ADD property, the channel delivers then to $p_j$ (at least)
one message from every sequence of $K$ consecutive messages sent by
$p_i$. Moreover, this message takes at most $D$ time units. This
means that at most $(K-1)$ messages can be lost (or take more than $D$
time units) between two messages from $p_i$ delivered consecutively by
$p_j$. As $p_i$ sends a message every $T$ clock ticks and the local
clocks run at a constant speed, the maximal delay between the consecutive
receptions by $p_j$ of messages sent by $p_i$ is $\Delta =(K-1)\times T+D$.
\end{proof}

\textbf{Lemma~\ref{onlyCorrect}}
\textit{Given a run $r$, there is a time $t^a$ after which there are no 
messages \ALIVE$(i,n-a)$ with $p_i \in \crashed(r)$ and $1 \leq a < n-1$.} \\

\begin{proof}
The proof of this lemma is by induction over $a$.

\noindent
\textit{Base case:} $a=1$. There is a time $t^1$ after which no
process sends \ALIVE$(i,n-1)$ messages with $p_i \in \crashed(r)$.

Let us remind that a generating message \ALIVE$(i,n-1)$ can be sent
only by process $p_i$. If $p_i$ crashes, it sends a finite number of
messages \ALIVE$(i,n-1)$. As this is true for any process that crashes,
there is a finite time $t^1$ after which generating messages are sent only
by correct processes.\\

\noindent
\textit{Induction case:}
Let us assume there is a time $t^a$ after which no
process sends messages \ALIVE$(i,n-a)$ with $p_i \in \crashed(r)$
and $1 \leq a < n-1$. To show that there is a time $t^{a+1}$ after
which no process sends \ALIVE$(i,n-(a+1))$ with $p_i \in \crashed(r)$
and $1< a+1 \leq n-1$, we consider two cases.
\begin{itemize}
%\vspace{-0.1cm}
\item 
Case 1: $a+1 \neq n-1$. Since channels neither create nor duplicate
messages and processes send a finite number of forwarding messages
\ALIVE$(i,n-a)$ before $t^a$ (as defined by induction assumption),
there is a finite time at which every process $p_j$ whose
$timer_j[i,n-a]$ is running, is such that $timer_j[i,n-a]$ expires for
the last time. When this occurs the predicate at line~\ref{ADD-20} is
evaluated. If the predicate is true (i.e. all the timers for $p_i$
expired), $p_j$ proposes itself as leader. Hence, the lemma follows
from the fact the next \ALIVE$\ $ message that $p_j$ sends, cannot be
\ALIVE$(i,-)$.

If the predicate at line~\ref{ADD-20} is not true, $p_j$ computes a
new hopbound value (with respect to $p_i$ if $p_j$ still considers it
as its current leader), which is the greatest hopbound value whose
timer has not expired and which has the lowest penalty number
(line~\ref{ADD-22}). It follows from the induction assumption that,
after time $t^a$, no process sends \ALIVE$(i,n-a)$ messages with
$1 \leq a < n-1$. Then, the new hopbound value (with respect to $p_i$)
must be at most $n-(a+1)$. So in the next \ALIVE$\ $ message, the
greatest hopbound value (with respect to $p_i$) that can be sent by
any process $p_j$ is $(n-(a+2))$, so no process $p_j$ sends forwarding
messages \ALIVE$(i,n-(a+1))$.
\item
Case 2: $a+1 = n-1$.
The proof of this case follows directly from the
predicate at line~\ref{ADD-07} (namely $hopbound_j[i] > 1$),
which prevents any process $p_j$ to send a message \ALIVE$(*,1)$. 
\end{itemize}
%\vspace{-0.5cm}
%\renewcommand{\toto}{onlyCorrect}
\end{proof}

\textbf{Theorem~\ref{main-theorem}} 
\textit{
 Given a run $r$ satisfying the {\em Span-Tree} property, there is a
 finite time after which the variables $leader_i$ of all the correct
 processes contain the smallest identity $\ell\in \correct(r)$.
 Moreover, after $p_\ell$ has been elected, there is a finite time after 
 which the only messages sent by processes are \ALIVE$(\ell,-)$ messages.\\}

\begin{proof}
 Initially (as any other process) the correct process $p_\ell$ with
 the smallest identity considers itself  leader (line~\ref{ADD-01}).
 Then it can be demoted only at line~\ref{ADD-12} when it receives a
 message \ALIVE$(j,-)$ such that $j<leader_\ell=\ell$
 (line~\ref{ADD-11}). As $p_\ell$ is the correct process with the
 smallest identity, it follows that such a message was sent by a
 faulty process $p_j$ (that crashed after it sent the generating
 message \ALIVE$(j,n-1)$). Due to Lemma~\ref{onlyCorrect},
 there is a finite time $\tau$ after which there are no more messages
 \ALIVE$(j,-)$ such that $j<\ell$. Hence, whatever the faulty
 process $p_j$, there is time $\tau'> \tau$ at which all the timers
 $timer_i[j,hb]$ with $hb\leq n-1$, have expired, and then $p_\ell$
 considers itself  leader (line~\ref{ADD-21}). Then, due to the
 predicate of line~\ref{ADD-11}, it can no longer be locally
 demoted. Moreover, due to Span-Tree assumption, there is a path made
 up of correct processes connected by $\diamond$ADD channels from
 $p_\ell$ to any other correct process. Due to
 Lemma~\ref{lemma:bounded-delay} it follows then that there is a
 finite time after which each correct process repeatedly receives
 messages \ALIVE$(\ell,-)$ with some hopbound value. Due to
 lines~\ref{ADD-11}-\ref{ADD-12}, processes adopts $p_\ell$ as
 leader. Since processes are repeatedly receiving messages
 \ALIVE$(\ell,-)$ with some hopbound value, the predicate at
 line~\ref{ADD-20} cannot become true as at least one hopbound value
 is always arriving on time at every correct process.

After $p_\ell$ has been elected, any alive process $p_i$ is such that
forever $leader_i=\ell$ and $\timer_i[\ell,-]$ for some hopbound value
never expires. It follows that, at line~\ref{ADD-08}, a process $p_i$
can send \ALIVE$(\ell,-)$, messages only.
\end{proof}

%----------------------------------------------------------------------------
\textbf{Lemma~\ref{lemmaUpperBound}} \textit{Let $p_i$ be a correct
process such that for every $t > t^r$, $hopbound_i[\ell] =
n-k$. Then, for every correct process $p_j$ such there is an
$\Diamond$ ADD channel from $p_i$ to $p_j$, $timeout_j[\ell,n-(k+1)]
\leq \mathcal{C} + 2^{log(\lceil \Delta\rceil )}$ with
$timeout_j[\ell,n-(k+1)]= \mathcal{C}$ before $t^r$.}		\\

\begin{proof}
 Since $p_i$ sends a message every $T$ units of time, by
 Lemma~\ref{lemma:bounded-delay}, after $t^r$, the maximum delay
 between the consecutive reception of two messages from $p_i$ to $p_j$
 is $\Delta$. After $t^r$, the $timeout_j[\ell,n-(k+1)]$ stops
 changing when $\Delta\leq timeout_i[\ell,n-(k+1)]$, so it cannot
 happen again that $p_j$ expires $timer_i[\ell,n-(k+1)]$. Therefore,
 the timeout is not incremented again. So in  the  case where
 $timeout_i[\ell,n-(k+1)] \geq \Delta$ before $t^r$, this lemma is
 true.
 
 In the other case, consider that $\Delta \leq \lceil \Delta \rceil$
 and then $2^{log(\Delta)} \leq 2^{log(\lceil \Delta \rceil)}$. Since
 the timeout increment under false suspicions is exponential,
 $timeout_i[\ell,n-(k+1)]$ needs to be incremented at most $\lceil log
 \ \Delta \rceil$ times for $\Delta\leq timeout_i[\ell,n-(k+1)]$. Once
 it is true that $\Delta\leq timeout_i[\ell,n-(k+1)]$, it cannot
 happen again that $p_i$ expires $timer_i[\ell,n-(k+1)]$, so the
 timeout is not incremented again. Therefore $timeout_i[\ell,n-(k+1)]
 \leq \mathcal{C} + 2^{log(\lceil \Delta \rceil)}$.

\end{proof}

%----------------------------------------------------------------------------

\textbf{Lemma~\ref{lemma:timedelta}} \textit{For every correct process
  $p_i$ such that there is a minimum length path of $\Diamond$~ADD
  channels of length $k$ from $p_\ell$ to $p_i$, $leader_i = \ell$ at
  time $t^c + (k \times \Delta)$. }\\

\begin{proof}
 Let $S^k = \{p_i|$ there is a minimum length path of $\Diamond$ ADD
 channels from $p_\ell$ to $p_i \}$ and let $x$ be the maximum length
 of a minimum length path of $\Diamond$ ADD channels from $p_\ell$ to
 any correct process.  The proof is by induction.
 
 \noindent \textit{Inductive base.}
 At time $t^c+\Delta$, every $p_i \in S^1$ has $leader_i = \ell$.
 
 Since after $t^c$ no timer expires again, by
 Lemma~\ref{lemma:bounded-delay}, after $\Delta$ units of time, every
 $p_i \in S^1$ receives a message from $p_\ell$ including itself as
 the leader. Since $p_\ell$ is the correct process with the smallest
 identity, condition in line~\ref{ADD-11} is true and $leader_i =
 \ell$ after $t^c + \Delta$.

 \noindent \textit{Inductive hypothesis.}  For $1\leq k < x$ and for
 every $p_i \in S^k$, $leader_i = \ell$ at time $t^c + (k \times
 \Delta)$.
 
 \noindent \textit{Inductive step.}  For $1< k+1 \leq x$ and for every
 $p_j \in S^{k+1}$, $leader_j = \ell$ at time $t^c + ((k+1) \times
 \Delta)$.
  
  Let $p_i \in S^k$. By inductive hypothesis, $leader_i = \ell$ at
  time $t^c + (k \times \Delta)$. Let $p_j \in S^{k+1}$ be a neighbor
  of $p_i$, namely, $p_j$ is connected to $p_i$ by an $\Diamond$ ADD
  channel and $k+1$ is the length of the minimum length path of
  $\Diamond$ ADD channels connecting $p_\ell$ to $p_j$. There can be
  two cases for $p_i$.
  
  \begin{enumerate}
   \item $hopbound_i[\ell] = n-k$. In this case, since the maximum
     distance between any two processes is $n-1$, condition in
     line~\ref{ADD-07} must be true and in $\Delta$ units of time at
     most, $p_j$ must receive an \ALIVE$(\ell, n-(k+1))$ message from
     $p_i$. Since $p_\ell$ is the correct process with the smallest
     identity, condition in line~\ref{ADD-11} is true and $leader_j =
     \ell$.
   \item $hopbound_i[\ell] = n-k'$ with $k'>k$. In this case, we have
     to show that $k' < (n-1)$ for the condition in line~\ref{ADD-07}
     to be true. The only way in which $p_i$ can have
     $hopbound_i[\ell] = n-k'$ is because there is a path of length
     $k'$ from $p_\ell$ to $p_i$. This path must be simple in order to
     $k' < (n-1)$, otherwise there must be a cycle in the path of
     length $n-1$ between $p_\ell$ and $p_i$. Let $p_\ell, q_1, q_2,
     ... , q_s, q_{s+1}, ..., q_c, q_s,..., p_i$ be that path with a
     cycle. So it must be that at some time $q_s$ has
     $hopbound_s[\ell] = n-s$ because it received \ALIVE$(\ell,n-s)$
     from $q_{s-1}$ and then $hopbound_s[\ell] = n-(s+a)$ with $a$ the
     length of the cycle that it received from $q_c$. But this cannot
     be forever since $q_{s-1}$ keeps sending \ALIVE$(\ell,n-s)$ to
     $q_s$, so eventually every $timer_m [\ell, n-(s-b)]$ with $ s < m
     \leq a$ and $s+1 \leq b \leq c$ for every process in the cycle
     must expire, which produces that $n-(s+a)$ does not arrive again
     and eventually $penalty_s[\ell,n-s] <
     penalty_s[\ell,n-(s+a)]$. Then, in line~\ref{ADD-16},
     $hopbound_s[\ell] = n-s$. So, this case cannot happen, and
     therefore $k' < (n-1)$, condition in line~\ref{ADD-07} must be
     true and in $\Delta$ units of time at most, $p_j$ must receive an
     \ALIVE$(\ell, n-(k'+1))$ message from $p_i$.
  \end{enumerate}
In both cases, $leader_i = \ell$ after $t^c + ((k+1) \times \Delta )$.

\end{proof}

%============================================================

% \textbf{Lemma~\ref{equal-set-lemma}}
% \textit{For any $p_i, p_j \in \correct(r)$ that are neighbors, eventually
% $known_i = known_j$.}\\

\begin{lemma}
\label{equal-set-lemma}
For any $p_i, p_j \in \correct(r)$ that are neighbors, eventually
$known_i = known_j$.
\end{lemma}

\begin{proof}
Let $p_i$ and $p_j$ be two correct neighbors. Assume that
$p_i$ communicates with $p_j$ through $\Diamond$ADD channel $a$ and
in the other direction is channel $b$. At the initialization, $p_i$
puts in all the $pending$ sets the pair $(\new,i)$ (line
\ref{RELS-04}). Since it sends this set every $T$ units of time
through all the channels (lines~\ref{RELS-07}-\ref{RELS-09}), for
Lemma~\ref{lemma:bounded-delay}, $p_j$ eventually receives the pending
set which contains at least the pair $(\new,i)$.
Then, for every pair $(label,k)$ that is received in $p_j$ there are
two cases: the pair contains new information or an acknowledgment.

When $p_j$ receives a pair $(\new,k)$ and  it is the first time that it
receives a pair with name $k$ (condition in line~\ref{RELS-17}), it
adds $k$ to its list of $known_j$ (line~\ref{RELS-18}), adds to every
pending set of every channel (except the one channel from which the
message arrived from) the pair $(\new,k)$ (line~\ref{RELS-20}).
In case in which $p_j$ already knows $k$, if $p_j$ is trying to send
this information as new to $p_i$, the pair is deleted from the pending
set because $p_i$ already knows this information (line~\ref{RELS-21}).
Finally, in either case $p_j$ adds the pair $(\ack, k)$ to the pending
set of channel $b$, namely, the one from which $p_j$ received the pair
$(\new,k)$ (line~\ref{RELS-23}).

The second case is when the pair is an acknowledgment $(\ack,
k)$. Since acknowledgment messages are only received if previously
$p_j$ sent $(\new,k)$ pair to $p_i$, then $p_j$ deletes the pair
$(\new,k)$ (if there is one) from the pending set of $b$ (line
~\ref{RELS-25}). Thus no pair with label $\new$ is deleted until an
acknowledgment is received or if it is received the same pair. But the
acknowledgment is added only if the 
$\new$ pair was received by the
receiver, meaning that $p_j$ knows the same names that $p_i$.

Thus, there cannot be some $k' \in known_i$ that $p_j$ does not know
eventually and viceversa.
\end{proof}

% \textbf{Lemma~\ref{correct-set-lemma}} \textit{For every $p_i,p_j \in \correct(r)$, such that $p_j$ is at distance $d$
% from $p_i$, there is a time $t^d$ after which $i \in  known_j$.}\\

\begin{lemma}%{lemma}{correctSetLemma}
\label{correct-set-lemma}
For every $p_i,p_j \in \correct(r)$, such that $p_j$ is at distance $d$
from $p_i$, there is a time $t^d$ after which $i \in  known_j$.
\end{lemma}
\begin{proof}
Let $d'$ be the maximum distance between $p_i$ and any other correct process.
The proof of this lemma is by induction over $d$ with $1 \leq d \leq
d'$.\\
{\it Base case}: $d=1$. There is a time $t^1$ after which for every
$j \in \correct(r)$ at distance $1$ from $p_i$, eventually $i \in
known_j$.

Since every process is connected by an $\Diamond$ ADD channel,
eventually every neighbor $p_j$ of $p_i$ receives a message from it
that contains in the pending set the pair $(\new,i)$ since $p_i$ added
to every pending set that pair initially (line~\ref{RELS-05}). Then,
every $p_j$ adds to its known set $i$, and adds to the pending set of
$p_i$ the pair $(\ack,i)$. Process $p_i$ does not delete pair
$(\new,i)$ from the pending set of $p_j$ till it receives an
$(\ack,i)$ from $p_j$ but this pair is only added if $p_j$ received
before the pair $(\new,i)$ from $p_i$ before. So eventually, every
neighbor knows $p_i$.

\noindent
{\it Induction case:}
Let us assume is there is a time $t^d$ after which for
every $p_j$ at distance $d < d'$ from $p_i$, $i \in known_j$.
We have to show that there is a time $t^{d+1}$ after which for
every $p_j$ at distance $d+1 \leq d'$ from $p_i$, $i \in known_j$

By the induction assumption, all processes $p_j$ at distance $d$ from
$p_i$ knows $i$. It means that before $p_j$ knew $p_i$ at some time
$p_j$ received the pair $(\new,i)$ from some neighbor. Since $p_j$ did
not knew $p_i$, it added the pair $(\new,i)$ to the pending set of
every neighbor. Then, eventually that pending set is sent to every
neighbor of $p_j$, so eventually $p_i$ is known by processes at
distance $d+1$.
\end{proof}

%------------------------------------------------------------------------

% \textbf{Lemma~\ref{onlyCorrectSn}} \textit{Given a run $r$, there is a finite time $t^a$ after which there are no
% messages \ALIVE$(i,k-a,pending)$ with $p_i \in \crashed(r)$ such that
% $1 \leq a < k \leq n$. }\\

\begin{lemma}%\begin{lemma}
\label{onlyCorrectSn}
Given a run $r$, there is a finite time $t^a$ after which there are no
messages \ALIVE$(i,k-a,pending)$ with $p_i \in \crashed(r)$ such that
$1 \leq a < k \leq n$. 
\end{lemma}

\begin{proof}
First, note that only process $p_i$ can change the entry
$hopbound_i[i]$  when it knows a new process
(line~\ref{RELS-18}). Since it only knew a finite number of processes
before it failed, then the entry $hopbound_i[i]$ is finite and has an
upper bound. Let us call $k$ the $hopbound_i[i]$ had before it failed.

The proof of this lemma is the same that
Lemma~\ref{onlyCorrect}, by strong induction over $a$. Just
note that when the timer expires is because \ALIVE$()$ message arrived
with $p_i$ as leader, no matter which hopbound of $p_i$ is expected.
  \end{proof}

Lemma~\ref{equal-set-lemma} shows  that eventually all the correct
processes forever have the same set $known$,  and 
Lemma~\ref{correct-set-lemma} proves that every correct process is
in the $known$ set of every correct process. Recall that for every process
$p_i$, entry $hopbound_i[i]$ is increased every time $p_i$ knows a new
process, namely, $hopbound_i[i]$ contains the cardinality of
$known_i$.  Then, we can conclude that there is a time $t^f$ after
which for every correct process $p_i$, there is a constant $k \leq n$
such that $hopbound_i[i] = k$, namely, each process has the same
hopbound. Let $\tau$ be the time after which all the failures already
happened. Let $\ell$ be the smallest identity in $\correct(r)$.

% 
% %------------------------------------------------------------------------
% \textbf{Lemma~\ref{levLeadership}}
% \textit{ Let $p_i \in \correct(r)$ at distance $d$ from $p_{\ell}$ after $t^f$
%  with $0 \leq d \leq n-1$. There is a time $ t > \tau$ and $t > t^f$
%  after which $hopbound_i[\ell] = k-d$ and $leader_i = \ell$ 
%  permanently.}\\
%  

\begin{lemma}
 \label{levLeadership}
 Let $p_i \in \correct(r)$ at distance $d$ from $p_{\ell}$ after $t^f$
 with $0 \leq d \leq n-1$. There is a time $ t > \tau$ and $t > t^f$
 after which $hopbound_i[\ell] = k-d$ and $leader_i = \ell$
 permanently.
\end{lemma}
\begin{proof}
We prove this lemma by strong induction on the length of the path
connecting $p_{\ell}$ to $p_i$.\\
\noindent
{\it Base case}: $d = 0$.  For process $p_\ell$, the two properties
are satisfied.\\ Due to Lemmas~\ref{equal-set-lemma}
and~\ref{correct-set-lemma}, is true that $hopbound_\ell[\ell] =
k$. It can not be changed by any process since the only processes that
can write $hopbound_\ell[\ell] = k$ is itself.

If $leader_\ell = \ell'$ such that $\ell'<\ell$, it means that
$p_{\ell'}$ failed since it is assumed that $\ell$ is the smallest
index of a correct process.  Due to Lemma \ref{onlyCorrectSn},
there is a time after which $p_\ell$ stops receiving messages with
leader $p_{\ell'}$. Eventually, the timer for $p_\ell'$ expires and
$p_\ell$ proposes itself as the leader. This can only happen a finite
number of times since the number of participating processes and the
number of messages sent are finite.

Let us consider the last time in which the timer for receiving a
message including a pair with $p_{\ell'}$ expires, then
line~\ref{RELS-39} is executed, so $leader_\ell = \ell$. Since
$p_\ell$ is the correct process with the minimum identity, it does not
execute again line~\ref{RELS-31}.

\noindent
 {\it Induction step:} Let us assume that there is a time after which
 $hopbound_j[\ell] = n-(m-1)$ and $leader_j = \ell$ permanently.  Let
 $p_i$ be a correct process at distance $m$ from $p_\ell$ and let $\pi
 = p_{\ell},...,p_j,p_i$ be a minimum length path of correct processes
 connecting $p_{\ell}$ to $p_i$. We have to sow that $hopbound_i[\ell]
 = k-m$ and $leader_i = \ell$ permanently.
 
By assumption induction there is a time after which correct processes
at distance $m-1$ satisfy the property. Since those processes sends
messages every $T$ units of time, eventually process $p_i$ receives a
message including $p_\ell$ as the leader.

If $leader_i = p_{\ell'}$ such that $\ell'<\ell$ eventually $leader_i
= p_\ell$ (same argument as base case).

Since $p_i$ and $p_j$ are connected by $\Diamond$ ADD channels,
eventually $p_i$ receives an \ALIVE\ message from $p_j$ including
$p_\ell$ as the leader, so it sets $leader_i = \ell$ because condition
in ~\ref{RELS-31}. Since there is a minimum length path of length $m$
connecting $p_\ell$ and $p_i$, $k-m$ is the greatest hopbound value of
$p_{\ell}$ that $p_i$ can receive. Then, the first condition in
line~\ref{RELS-32} is true and the $hopbound_i[\ell]$ is set to $k-m$.

If another process sends an \ALIVE\ message including a pair
$(p_{\ell}, m', pending)$ with $m < m'$ we have two cases.
\noindent
\begin{itemize}
  \item Case 1: The timer for $p_{\ell}$ is expired. In that case,
second condition in line \ref{RELS-32} is true, so $hopbound_i[\ell] =
m$ (line~\ref{RELS-33}) and the timeout is increased
(line~\ref{RELS-35}). Since processes are connected by $\Diamond$ ADD
channels, there is a time after which messages from $p_j$ arrives in
at most $\Delta$ to $p_i$, meaning that if the timer is expired, every
time that $p_i$ receives an \ALIVE\ message from $p_j$ it increases
the timeout. Eventually, the timeout gets a value $X > \Delta$ and
stops changing.

This means that eventually this case is no longer true, namely the
$timer_i[\ell]$ does not expire and as consequence $hopbound_i[\ell] =
k-m$ and $leader_i = \ell$ does not change.
% %\vspace{-0.1cm}
  \item 
\textit{Case 2:} The timer for $p_{\ell}$ is not expired, so only
first condition in line \ref{RELS-34} can be true, since by induction
assumption, $p_j$ sends $(p_\ell,n-m,pending)$ to $p_i$ every $T$
units of time, so the $hopbound_i[\ell] = k-m$ does not change and
since we assumed that $\ell$ is the minimum index number, it must be
that $leader_i = \ell$.
\end{itemize}
%\vspace{-0.2cm}
%%\renewcommand{\toto}{levLeadership}
 \end{proof}

%------------------------------------------------------------------------

% \textbf{Lemma~\ref{cost-lemma}} \textit{There is a time after which,
%   at each correct process $p_i$ and any of its channels $m$, the set
%   $pending_i[m]$ becomes and forever remains empty.} \\

\begin{lemma}
\label{cost-lemma}
There is a time after which, at each correct process $p_i$ and any of its
channels $m$, the set $pending_i[m]$ becomes and  forever remains empty.
\end{lemma}
\begin{proof}
 Let $p_i$ and $p_j$ be two correct neighboring processes such that
 $p_i$ is connected to $p_j$ through channel $a$ and in the other
 direction through channel $b$. We want to show that eventually
 $pending_i[a]$ ($pending_i[b]$) is empty.
 
 Assume, without loss of generality, that $p_i$ adds pair $(\new,k)$
 to $pending_i[a]$. Eventually $p_j$ receives the pair $(\new,k)$ from
 $p_i$. If this pair is already in $pending_j[b]$, then it is deleted
 (line~\ref{RELS-21}) because $p_i$ already knows $k$. Then $p_j$ adds
 to $pending_j[b]$ pair $(\ack,k)$ (line ~\ref{RELS-23}).
 
Eventually, $p_i$ receives the pair $(\ack,k)$ from $p_j$, so it
deletes its $(\new,k)$ pair from $pending_i[a]$
(line~\ref{RELS-25}). Then, eventually $p_j$ receives a message from
$p_i$ without the pair $(\new,k)$, so $p_j$ deletes the $(\ack,k)$ from
$pending_j[b]$ (line~\ref{RELS-28}).

 Process $p_i$ can only add a finite number of $\new$ pairs in
 $pending_i[a]$, since the number of different processes 
 is finite. Then, $p_j$ can only add an $\ack$ pair to $pending_j[b]$
 if it receives a $\new$ pair from $p_i$, namely, the number of $\ack$
 pairs that it can add to $pending_j[b]$ is finite too.
 
 All the $\new$ pairs are deleted as soon as the acknowledgment
 arrives, and the $\ack$ pairs are deleted as soon as the $\new$ pairs
 stops arriving. So eventually, every pending set of correct processes
 is empty.
\end{proof}

\begin{theorem}
\label{main-theorem-UNK}
 Given a run $r$, there is a finite time after which the variables
 $leader_i$ of all the correct processes contain forever the smallest
 identity $\ell\in \correct(r)$.
\end{theorem}

\begin{proof}
The proof follows from Lemma~\ref{levLeadership}.
\end{proof}

\begin{theorem}
\label{cost-theorem-UNK}
Eventually, the size of a message is $O(\llog~n)$. 
\end{theorem}

\begin{proof}
  By Lemma~\ref{cost-lemma}, eventually each variable $pending_i[m]$
  of a correct process is forever empty. 
  So eventually, any \ALIVE\ message carries a process identity
which belongs to the set $\{1,\cdots,n\}$ and a hopbound number
$hopbound$ such that $2\leq hopbound \leq n-1$.
\end{proof}

\end{document}